\documentclass[12pt,reqno]{amsart}

\usepackage{amsmath,amsfonts,amsbsy,amsthm} %
\usepackage{eucal}
\usepackage{dsfont}

\usepackage{enumerate}
\usepackage{enumitem}
\setlist{leftmargin=*}

\numberwithin{equation}{section}

\makeatletter
\newtheoremstyle{corsivo}
   {\medskipamount}{\medskipamount}%
   {\itshape}{}%
   {\bfseries}{}%
   { }
   {\thmname{#1}\thmnumber{\@ifnotempty{#1}{ }\@upn{#2}}%
    \thmnote{ {\bfseries(#3)}}.}%
\makeatother

\theoremstyle{plain}
\newtheorem{thm}{Theorem}
\newtheorem{lemma}[thm]{Lemma}

\makeatletter
\newtheoremstyle{dritto}
   {\medskipamount}{\medskipamount}%
   {\rmfamily}{}%
   {\bfseries}{}%
   { }
   {\thmname{#1}\thmnumber{\@ifnotempty{#1}{ }\@upn{#2}}%
    \thmnote{ {\bfseries(#3)}}.}%
\makeatother

\theoremstyle{dritto}
\newtheorem{dfn}{Definition} 
\newtheorem{rmk}[thm]{Remark}

\newtheorem{assumption}{Assumption}    

\newcommand{\sub}[1]{_{\mathrm{#1}}}

\newcommand{\eps}{\varepsilon}
\newcommand{\epsi}{\varepsilon}

\newcommand{\eu}{\mathrm{e}}
\newcommand{\iu}{\mathrm{i}}   

\newcommand{\N}{\mathbb{N}}

\newcommand{\R}{\mathbb{R}}
\newcommand{\C}{\mathbb{C}}

\newcommand{\Do}{\mathcal{D}}

\newcommand{\Hi}{\mathcal{H}}

\newcommand{\bra}[1]{\left\langle #1 \right|}
\newcommand{\ket}[1]{\left| #1 \right\rangle}

\newcommand{\set}[1]{\left\lbrace#1\right\rbrace}    

\newcommand{\ie}{{\sl i.\,e.\ }}   
\newcommand{\eg}{{\sl e.\,g.\ }} 
\newcommand{\etal}{{\sl et al.\ }}  

\newcommand{\virg}[1]{``#1''}
\newcommand{\crucial}[1]{{\textbf{#1}}}    

\renewcommand{\(}{\left(}
\renewcommand{\)}{\right)}

\renewcommand{\endrmk}{\hfill $\diamond$}

\usepackage{color}

\let\oldfootnote\footnote
\renewcommand{\footnote}[1]{\oldfootnote{\  #1}}

\usepackage{amssymb}
\usepackage{bbm}

\usepackage{stackrel}
\usepackage{braket}
\usepackage{hyperref}
\usepackage[title,titletoc,toc]{appendix}
\usepackage{tikz}

\newcommand{\1}{\mathbbm{1}}

\newcommand{\fhalf}{{\frac{1}{2}}}

\newcommand{\bm}[1]{\boldsymbol{#1}}
\newcommand{\Sch}{Schr\"odinger }

\newcommand{\abs}[1]{\left| #1 \right|}
\newcommand{\norm}[1]{\left\| #1 \right\|}
\newcommand{\inner}[2]{\left\langle  #1 \, , \,  #2 \right\rangle}

\newcommand{\upk}[1]{\ket{#1}\otimes e_{1}}
\newcommand{\dwnk}[1]{\ket{#1}\otimes e_{-1}}

\theoremstyle{plain}

\title[Controllability of the Jaynes-Cummings dynamics]{Approximate controllability of \\[1mm] the Jaynes-Cummings dynamics}
\author[L. Pinna and G. Panati]{Lorenzo Pinna and Gianluca Panati}

\date{December 15, 2017. Version submitted to \textsl{arXive.org}}

\begin{document}
\maketitle
\begin{abstract}
We investigate the controllability of the Jaynes-Cum\-mings dynamics in the resonant and nearly resonant regime. 
We analyze two different types of control operators acting on the bosonic part, corresponding 
-- in the application to cavity QED -- to an external electric and magnetic field, respectively. \newline
\noindent We prove approximate controllability for these models, for all values of the coupling constant $g \in \R$ except those in a countable set $S_*$ which is explicitly characterized in the statement. The proof relies on a spectral analysis which yields the non-resonance of the spectrum for every $g \in \R \setminus S_*$. 

\vspace{1mm}

\end{abstract}

\section{Introduction}
\label{Sec:Intro}

Spin-boson models, which describe the interaction between a $2$-level quantum system and finitely many distinguished modes of a bosonic field, are ubiquitous in Quantum Theory. They play a prominent role in quantum optics, magnetic resonance theory and in cavity Quantum Electro Dynamics (QED). In the latter context, they provide an approximate yet accurate description of the dynamics of a $2$-level atom in a resonant microwave cavity, as in recent experiments \cite{BRH, HaRa}. Among the spin-boson models, two prototypical examples are the Rabi model  \cite{Rab1,Rab2}  and the Jaynes-Cummings model \cite{JC}, which despite their age are still very popular in several fields. 


More recently, these models attracted the attention of researchers in mathematical control theory.  In a nutshell, the controllability 
problem is the following. An initial state $\Psi\sub{in}$ and a target state $\Psi\sub{fin}$ are given, as well as the unperturbed Hamiltonian operator $H\sub{up}$ and the control operator $H\sub{con}$. Then, one investigates whether there exist a number $T >0$ and a reasonably regular (\eg piecewise constant) function $u: [0,T] \rightarrow \R$ such that the dynamics generated by the Hamiltonian operator
$$
H(t) := H\sub{up} + u(t) H\sub{con}
$$
drives the state $\Psi(0) = \Psi\sub{in}$ so that $\Psi(T)$ is (in an $\epsi$-neighbourod of) the state $\Psi\sub{fin}$. 
If the latter holds true for every normalized $\Psi\sub{in}$, $\Psi\sub{fin}$ (and every $\epsi > 0$) the system is called 
\crucial{(approximately) controllable}  (see Definition \ref{Def:Controllability} for details). 
It has been noticed that, whenever the Hilbert space is infinite dimensional, controllability is generally false, hence one focuses on approximate controllability \cite{BMS,Tur}.

In the field of quantum control, a crucial role is played by the competition between symmetry and controllability.  
In general, symmetries are an obstacle to controllability, because they imply the existence of invariant subspaces for the  dynamics. Therefore, the external control operator must necessarily break all the symmetries of the unperturbed Hamiltonian in order to achieve approximate controllability. There are highly symmetric systems that are not approximately controllable, 
as proved by Mirrahimi and Rouchon for the $1$-dimensional harmonic oscillator \cite{uncontr}. 
On the other hand, approximate controllability has been proved for trapped ion models \cite{EbLaw, ErvPuel, PaduroSigalotti} and, more recently, for the Rabi Hamiltonian \cite{BMPS}. 
So, one may wonder whether a model which is somehow intermediate between the harmonic oscillator and the Rabi Hamiltonian is still approximately controllable.  
In particular the Jaynes-Cummings (JC) model, despite its superficial similarity with the Rabi model, has an additional symmetry  with respect to the latter, corresponding to the conservation of the total number of excitations.  Hence, its controllability is an interesting matter in control theory, and the question whether the JC dynamics is controllable has been considered
by Rouchon some years ago  \cite[Section 4.2]{Rou}. 


In this paper, we answer the previous question. We consider two different types of control operator, which in the application to cavity QED correspond, respectively, to an external magnetic and an external electric field.  (For the connection between the JC model and the standard model of non-relativistic QED, namely the {\it Pauli-Fierz model}, we refer to  \eg \cite[Section I.A]{BMPS} and references therein). 
In Theorem \ref{short} we prove that the Jaynes-Cummings dynamics is controllable for every value of the coupling constant with the exception of a set $S_*$ of measure zero. Then, in Theorem \ref{long} we characterize the values in $S_*$ as solutions to some explicit equations. The proof exploits three technical ingredients: the integrability of the JC model \cite{JC}; a controllability criterion proved by Boscain \etal \cite{BCCS}, which is based on the study of the resonances of the spectrum; a detailed analysis of the resonance condition.

As far as future perspectives are concerned, an interesting task is to provide a constructive control method for the Jaynes-Cummings and the Rabi dynamics. Notice that in this paper we make an explicit  construction of a non-resonant chain of connectedness (see Definition \ref{connchain}). However, as far as we know, this fact implies the approximate controllability of the system only via a general theorem \cite{BCCS} whose proof is not constructive. Finally, we mention a related interesting problem. It is known that the Jaynes-Cummings dynamics can be heuristically seen as an approximation of the Rabi dynamics, in an appropriate regime, as discussed in detail in \cite[Section 4.2]{Rou}. A rigorous mathematical proof of the latter claim, which is still missing in the literature, would provide a deeper understanding of both the models and their dynamics.  

\medskip

\noindent {\bf Acknowledgments.}  We are grateful to U.\,Boscain, M.\,Sigalotti,  P.\,Ma\-son,  and S.\,Teufel 
for stimulating and useful discussions on related topics. We also thank M.\,Moscolari for a careful reading of the manuscript.


\section{The Jaynes-Cummings model}
\subsection{Definition of the model}
In the Hilbert space $\Hi=L^2(\R)\otimes\C^2$ we consider the \Sch equation
\[ \iu \hbar \, \partial_t\psi=H_{\rm JC}\psi \]
with Hamiltonian operator  (JC Hamiltonian)
\begin{equation}\label{JCH}
H_{\rm JC}\equiv H_{\rm JC}(g)=\frac{\hbar\omega}{2}(X^2+P^2)\otimes\1+\frac{\hbar\Omega}{2}\1\otimes\sigma_z+\frac{\hbar g}{\sqrt{2}}(X\otimes\sigma_x-P\otimes\sigma_y)
\end{equation}
where $\omega,\Omega\in\R_+$ and $g\in\R$ are constants, $X$ is the position operator, \ie $X\psi(x)=x\psi(x)$, and $P=-\iu\partial_x$. The operators $\sigma_x,\sigma_y,\sigma_z$ acting on $\C^2$ are given by the Pauli matrices 
\[ \sigma_x=\left(\begin{array}{cc}0&1 \\ 1&0\end{array}\right)\quad
   \sigma_y=\left(\begin{array}{cc}0&-\iu \\ \iu&0\end{array}\right)\quad
   \sigma_z=\left(\begin{array}{cc}1&0 \\ 0&-1\end{array}\right).
\]
The quantity $\Delta:=\Omega-\omega$ is called \emph{detuning} and measures the difference between the energy quanta of the two subsystems corresponding to the factorization of the Hilbert space.

By introducing the \emph{creation} and \emph{annihilation} operators for the harmonic oscillator, defined as usual by 
\begin{equation}\label{adagga}
a^\dagger=\frac{1}{\sqrt{2}}(X-\iu P)\qquad a=\frac{1}{\sqrt{2}}(X+\iu P),
\end{equation}
and the \emph{lowering} and \emph{raising} operators 
\begin{equation}\label{low_raise}
\sigma=\frac{1}{2}(\sigma_x-\iu\sigma_y)=\left(\begin{array}{cc}0&0\\1&0\end{array}\right)\quad\sigma^\dagger=\frac{1}{2}(\sigma_x+\iu\sigma_y)=\left(\begin{array}{cc}0&1\\0&0\end{array}\right), 
\end{equation}
the JC Hamiltonian (omitting tensors) reads
\[ H_{\rm JC}=\hbar\omega \(a^\dagger a +\fhalf\)+\frac{\hbar\Omega}{2}\sigma_z
+ \frac{\hbar g}{2} \(a\sigma^\dagger+a^\dagger\sigma \). \]

The popularity of this model relies on the fact that it is presumably the simplest model describing a two-level system interacting with a distinguished mode of a quantized bosonic field (the harmonic oscillator). 
It was introduced by Jaynes and Cummings in 1963 as an approximation to the Rabi Hamitonian
\begin{equation}\label{RaH}
H_{\rm R}=H_{\rm R}(g)=\hbar\omega \(a^\dagger a +\fhalf \)+\frac{\hbar\Omega}{2}\sigma_z+\frac{\hbar g}{2}(a+a^\dagger)
(\sigma+\sigma^\dagger).
\end{equation}
The latter traces back to the early works of Rabi on spin-boson interactions  \cite{Rab1,Rab2}, while in \cite{JC} Jaynes and Cummings derived both \eqref{JCH} and \eqref{RaH} from a more fundamental model of non-relativistic Quantum Electro Dynamics (QED).

Nowadays, both Hamiltonians \eqref{JCH} and \eqref{RaH} are widely used in several fields of physics. Among them, one of the most interesting is cavity QED. In typical cavity QED experiments, atoms move across a cavity that stores a mode of a quantized electromagnetic field. During their passage in the cavity the atoms interact with the field: the Hamiltonians \eqref{JCH} and \eqref{RaH} aim to describe the interaction between the atom and the cavity, in different regimes \cite{BRH, HaRa}.  
More precisely, \eqref{JCH} and \eqref{RaH} can be heuristically derived from a mathematical model of non-relativistic QED, the Pauli-Fierz model \cite{Spohn}; we refer to \cite{Coh1} and the more recent \cite{BMPS} for a discussion of this derivation. 
 
The approximation consisting in replacing \eqref{RaH} with \eqref{JCH} is commonly known as  the \emph{rotating wave approximation} (or \emph{secular approximation}), and is valid under the assumptions \cite{Rou}
\begin{equation}
\abs{\Delta}\ll\omega,\Omega\qquad g\ll\omega,\Omega
\end{equation}
which mean that the harmonic oscillator and the two-level system are almost in resonance and the coupling strength is small compared to the typical energy scale. Heuristically, in this regime the probability of creating or destroying two excitations is negligible, thus one can remove the so-called \emph{counter-rotating} terms $a^\dagger\sigma^\dagger$ and $\,a\sigma$ in \eqref{RaH} to obtain \eqref{JCH}. 
More precisely, the justification of this approximation relies on separation of time scales, a well-know phenomenon in several areas of physics \cite{PST1,PST2,PSparberT}. Indeed, by rewriting the dynamics generated by \eqref{RaH} in the interaction picture with respect to 
\begin{equation} \label{H_0}
H_0:=H_{\rm JC}(0)=H_{\rm R}(0)= \hbar \omega \(a^\dagger a + \fhalf\) +\frac{\hbar \Omega}{2}\sigma_z ,
\end{equation}
one gets
\begin{align}
\nonumber e^{\iu H_0t/\hbar}(H_{\rm R} - H_0)e^{-\iu H_0t/\hbar} 
&=\frac{g}{2} \(e^{-\iu (\Omega-\omega)t}a^\dagger \sigma + e^{\iu (\Omega-\omega)t}a\sigma^\dagger \) \\
& +\frac{g}{2} \(e^{-\iu (\Omega+\omega)t}a\sigma+e^{\iu (\Omega+\omega)t}a^\dagger\sigma^\dagger \). 
\end{align}
One notices that the terms $a^\dagger\sigma,\,a\sigma^\dagger$ oscillate with frequency $\abs{\omega-\Omega}$, 
while $a^\dagger\sigma^\dagger,\,a\sigma$ oscillate on the faster scale $\omega+\Omega$, so that 
the latter average to zero on the long time scale $\abs{\omega-\Omega}^{-1}$. 
While the physical principles leading from \eqref{RaH} to \eqref{JCH} are clear, as we mentioned in the Introduction 
a rigorous mathematical justification for this approximation seems absent from the literature, as recently remarked in \cite{Rou}. 

\vspace{1mm}

We use hereafter Hartree units, so that in particular $\hbar=1$. 

\subsection{Spectrum of the JC Hamiltonian}
While apparently similar, the JC Hamiltonian \eqref{JCH} and the Rabi Hamiltonian \eqref{RaH} are considerably different from the viewpoint of symmetries. 

As operators, they are both infinitesimally small perturbation, in the sense of Kato \cite{Ka}, of the free Hamiltonian $H_0$ (defined in \eqref{H_0}), which has compact resolvent. Eigenvalues and eigenvectors of $H_0$ are easily obtained by tensorization, starting from the eigenvectors $\{e_1,\,e_{-1}\}$ of $\sigma_z$ and the standard basis of $L^2(\R)$ given by  real eigenfunctions of $a^\dagger a$, namely the Hermite functions
\[ \ket{n}=\frac{1}{\sqrt{2^{n}n!\sqrt{\pi}}} \, h_n(x) \, e^{-\frac{x^2}{2}},\qquad n\in\N, \]
where $h_n$ is the $n$-th Hermite polynomial. As well known, they satisfy
\begin{equation}\label{annih_cre}
a^\dagger a\ket{n}=n\ket{n},\qquad a^\dagger\ket{n}=\sqrt{n+1}\ket{n+1},\qquad a\ket{n}=\sqrt{n}\ket{n-1}.
\end{equation}
Then 
\[ H_0\upk{n}=E^0_{(n,1)}\upk{n},\quad H_0\dwnk{n}=E^0_{(n,-1)}\dwnk{n} \]
with
\[ E^0_{(n,s)}=\omega(n+\fhalf)+s\frac{\Omega}{2},\qquad n\in\N,\;s\in\lbrace -1,1\rbrace. \]
Since $(a+a^\dagger)\sigma_x$ and $(a\sigma^\dagger+a^\dagger\sigma)$ are infinitesimally $H_0$-bounded,
by standard perturbation theory $\lbrace H_{\rm JC}(g)\rbrace_{g\in\C}$ and $\lbrace H_{\rm R}(g)\rbrace_{g\in\C}$
are analytic families (of type A) of operators with compact resolvent \cite[Section VII.2]{Ka}. Therefore, by Kato-Rellich theorem, the eingenvalues and eigenvectors of $H_{\rm JC}(g)$ and $H_{\rm R}(g)$ are analytic functions of the parameter $g$. Coefficients of the series expansion of eingenvalues and eigenvectors can be explicitly computed \cite{RS4}.

From the viewpoint of symmetries, it is crucial to notice that, as compared to the Rabi Hamiltonian, the JC Hamiltonian has an additional conserved quantity, namely the total number of excitations, represented by the operator $C=a^\dagger a+\sigma^\dagger\sigma$. 
As a consequence, the JC Hamiltonian reduces to the invariant subspaces
\begin{equation}
\Hi_{n}=\mathrm{Span}\{\upk{n}, \dwnk{n+1}\}\quad n\geq 0,\qquad\Hi_{-1}=\mathrm{Span}\{\dwnk{0}\},
\end{equation}
which are the subspaces corresponding to a fixed number of total excitations, \ie $ C \upharpoonright_{\Hi_n}=n+1$. 
Indeed,  $H_{\rm JC}$ restricted to these subspaces reads
\begin{align}
\nonumber H_n(g)& :=H_{\rm JC}(g)\upharpoonright_{\Hi_n}=\left(
\begin{array}{cc}
E^0_{(n,1)} & g\sqrt{n+1}\\
g\sqrt{n+1} & E^0_{(n+1,-1)}\\
\end{array}\right)\\
&=\omega(n+1)\1+\left(
\begin{array}{cc}
\Delta/2 & g\sqrt{n+1}\\
g\sqrt{n+1} & -\Delta/2\\
\end{array}
\right).
\end{align}
Eigenvalues and eigenvectors of $H_n$ are easily computed to be
\begin{equation}
H_{\rm JC}(g)\ket{n,\nu}=E_{(n,\nu)}\ket{n,\nu},\qquad\qquad n\in\N,\;\nu\in\lbrace -,+\rbrace
\end{equation}
where
\begin{align}
\label{eig} E_{(n,\nu)}(g) & =\omega(n+1)+\nu\frac{1}{2}\sqrt{\Delta^2+4g^2(n+1)}\\
\label{eig_n+}\ket{n,+}(g) & =\cos(\theta_n/2)\upk{n}+\sin(\theta_n/2)\dwnk{n+1}\\
\label{eig_n-}\ket{n,-}(g) & =-\sin(\theta_n/2)\upk{n}+\cos(\theta_n/2)\dwnk{n+1}
\end{align}
and the \emph{mixing angle} $\theta_n(g)\in [-\pi/2,\pi/2]$ is defined through the relation
\begin{equation}\label{theta}
\tan\theta_n:=\frac{2g\sqrt{n+1}}{\omega-\Omega}.
\end{equation}
Hereafter, we will  omit the $g$-dependence of the eigenvectors $\ket{n,\nu}$ for the sake of a lighter notation. 
Observe that in the resonant case, \ie $\Delta=0$, equation \eqref{theta} implies $\abs{\theta_n}=\pi/2$ for every $n\in\N$, hence the eigenvectors $\ket{n,\nu}$ are independent from $g$,  while the eingenvalues still depend on it. 

Moreover, depending on the sign of $\Delta$, one has 
\[ \begin{array}{llr}
E_{(n,+)}(0)=E^0_{(n,1)} & E_{(n,-)}(0)=E^0_{(n+1,-1)}, & \qquad\mbox{for }\Delta>0 \\[2mm]
E_{(n,+)}(0)=E^0_{(n+1,-1)} & E_{(n,-)}(0)=E^0_{(n,1)}, & \qquad\mbox{for }\Delta<0\\[2mm]
E_{(n,\nu)}(0)=E^0_{(n+1,-1)}=E^0_{(n,1)}, & & \qquad\mbox{for }\Delta=0
\end{array}. \]
As we mentioned before, in view of Kato-Rellich theorem, the eigenvalues of $H\sub{JC}(g)$ are analytic in $g$ if a convenient labeling is chosen. The table above shows which function, among $g \mapsto E_{(n,+)}(g)$ and $g \mapsto E_{(n,-)}(g)$, provides the analytic continuations of the spectrum at the points $E^0_{(n,1)}$ or $E^0_{(n+1,-1)}$. 
When $\Delta=0$, in order to have analytic eigenvalues and eigenfunctions we must choose $E_{(n,\nu)}=\omega(n+1)+\nu\sqrt{n+1}g$.

The spurious eigenvector $\dwnk{0}$ with eigenvalue $E^0_{(0,-1)}=\Delta/2$ completes the spectrum of the JCH. Let us define
\[ {\delta\equiv}\delta(\Delta):=\left\lbrace
	\begin{array}{lr}
	+ & \mbox{if }\Delta\geq 0\\
	- & \mbox{if }\Delta<0
	\end{array}.\right.
\]
Throughout the paper we will use the notation $\ket{-1,\delta}:=\dwnk{0}$ and $E_{(-1,\delta)}:=E^0_{(0,-1)}$. We will denote a pair $(n,\nu)$ with a bold letter $\bm{n}$, meaning that the first component of $\bm{n}$ is the same not-bold letter while the second component is the corresponding Greek letter, namely
\[ \bm{n}=(n,\nu),\quad\bm{n}(1)=n,\quad\bm{n}(2)=\nu. \]
Let also us define 
\begin{equation}\label{fns}
f_n(g):=\fhalf\sqrt{\Delta^2+4g^2(n+1)}.
\end{equation}
With this notation, we can write the specrum of the JC Hamiltonian in a synthetic way as
\begin{equation} \label{Eq:Spectrum}
\sigma\Bigl(H_{\rm JC}(g)\Bigr)=\lbrace E_{\bm{n}}\rbrace_{\bm{n}\in\mathcal{N}},\qquad E_{\bm{n}}(g)=\omega(n+1)+\nu f_n(g) 
\end{equation}
where
\begin{equation}\label{indset}
\mathcal{N}:=\left(\N\times\lbrace-,+\rbrace\right)\cup\lbrace(-1,\delta(\Delta))\rbrace.
\end{equation}
Notice that the notation is coherent since 
$$
E_{(-1,\delta)}=\delta(\Delta)f_{-1}(g)=\delta(\Delta)\frac{\abs{\Delta}}{2}=\frac{\Delta}{2}=E^0_{(0,-1)}, 
$$
in agreement with the definition above.
It will be also useful introduce the following sets
\begin{equation}\label{newN}
\mathfrak{N}_\pm:=\N\cup\{\mp\delta(\Delta)1\}
\end{equation}
which are copies of the natural numbers with $\set{-1}$ added to the set with the index $\delta(\Delta)$. 


\goodbreak

\section{General setting and main result}
\subsection{General setting}
We now introduce the controllability problem in a general setting. Let $\Hi$ be a separable Hilbert space 
with Hermitian product $\inner{\cdot}{\cdot}$. 
We consider the equation
\begin{equation}\label{contrsys}
\partial_t\psi=(A+u(t)B)\psi,\quad\psi\in\Hi
\end{equation}
where $A,B$ are skew-adjoint linear operator on $\Hi$ with domain $\Do(A)$ and $\Do(B)$ respectively, $u$ is a function of time with values in $U\subset\R$.
\begin{assumption}\label{ass1}
The system $(A, B, U, \Phi_\mathcal{I})$ is such that:
\begin{enumerate}[label=$(\mathrm{A}_{\arabic*})$,ref=$(\mathrm{A}_{\arabic*})$]
\item \label{item:A1} 
$\Phi_\mathcal{I}=\{\phi_k\}_{k\in\mathcal{I}}$ is a Hilbert basis of eigenvectors for $A$ associated to the eigenvalues $\{\iu\lambda_k\}_{k\in\mathcal{I}}$;
\item \label{item:A2}$\phi_k\in\Do(B)$ for every $k\in\mathcal{I}$;
\item \label{item:A3}$A+wB : \mathrm{Span}_{k\in\mathcal{I}}\{\phi_k\}\rightarrow\Hi$ is essentially skew-adjoint for every $w\in U$;
\item \label{item:A4} if $j\neq k$ and $\lambda_j=\lambda_k$, then $\inner{\phi_j}{B\phi_k}=0$. \endrmk
\end{enumerate} 
\end{assumption}

\goodbreak

Under these assumptions $A+uB$ generates a unitary group $t\mapsto \eu^{(A+uB)t}$ for every constant $u\in U$. Hence,  for every piecewise constant function  $u(t)=\sum_{i=1}^n u_i\chi_{[t_{i-1},t_i]}(t)$ associated to a partition $0=t_0<t_1< \ldots < t_n$, we can define the propagator
\begin{equation}\label{propa}
\Upsilon_t^u:=\eu^{(t-t_j)(A+u_{j+1}B)} \eu^{(t_{j}-t_{j-1})(A+u_{j}B)} \dots \eu^{t_1(A+u_1B)}\quad\mbox{for}\quad t_j<t\leq t_{j+1}.
\end{equation}
The solution to \eqref{contrsys} with initial datum $\psi(0)=\psi_0\in\Hi$ is denoted by $\psi(t)=: \Upsilon_t^u{(\psi_0)}$.
\begin{dfn}\label{Def:Controllability} 
Let $(A, B, U, \Phi_\mathcal{I})$ satisfy Assumption \ref{ass1}. We say that \eqref{contrsys} is \textbf{approximately controllable} if for every $\Psi\sub{in}$ and $\Psi\sub{fin}$  with $\norm{\Psi\sub{in}}=1=\norm{\Psi\sub{fin}}$ and for every $\eps>0$ there exist a finite $T_\eps>0$ and a piecewice constant control function $u : [0,T_\eps]\rightarrow U$ such that 
\[ \norm{\Psi\sub{fin} - \Upsilon_{T_\eps}^u(\Psi\sub{in})}<\eps. \] \endrmk
\end{dfn}
We recall a criterion for approximate controllability on which our proof is based. This general result gives a sufficient condition for approximate controllability based on the spectrum of $A$ and the action of the control operator $B$. More precisely, if $\sigma(A)$ has a sufficiently large number of non-resonant transitions, \ie pairs of levels $(i,j)$ such that their energy difference $\abs{\lambda_i-\lambda_j}$ is not replicated by any other pair, and $B$ is able to activate these transitions, then the system is approximately controllable. This heuristic idea is made precise in the following 

\begin{dfn}\label{connchain}
Let $(A, B, U, \Phi_\mathcal{I})$ satisfy Assumption \ref{ass1}. A subset $\mathcal{C}$ of the set $\mathcal{I} \times \mathcal{I} = \mathcal{I}^2$ \textbf{connects a pair} $(j,k)\in\mathcal{I}^2$, if there exists a finite sequence $c_0,...,c_p$ such that:
\begin{enumerate}[label=$({\roman*})$,ref=$({\roman*})$]
\item[(i)] $c_0=j$ and $c_p=k$;
\item[(ii)] $(c_i,c_{i+1})\in\mathcal{C}$ for every $0\leq i\leq p-1$;
\item[(iii)] $\inner{\phi_{c_i}}{B\phi_{c_{i+1}}}\neq 0$ for every $0\leq i\leq p-1$.
\end{enumerate}
The set $\mathcal{C}$ is called a \textbf{chain of connectedness} for $(A, B, U, \Phi_\mathcal{I})$ if connects every pair in $\mathcal{I}^2$.\\
A \textbf{chain of connectedness} is called \textbf{non-resonant} if for every $(c_1,c_2)\in \mathcal{C}$ holds
\[ \abs{\lambda_{c_1}-\lambda_{c_2}}\neq\abs{\lambda_{t_1}-\lambda_{t_2}} \] 
for every $(t_1,t_2)\in\mathcal{I}^2\setminus\{(c_1,c_2),(c_2,c_1)\}$ such that $\inner{\phi_{t_1}}{B\phi_{t_2}}\neq 0$.
\endrmk
\end{dfn}
Intuitively, if two levels of the spectrum are non-resonant and the control operator $B$ couples them, one can tune the control function $u$ in such a way to arrange arbitrarily the wavefunction's components on these levels, without modifying any other component. Therefore, having a non-resonant chain of connectedness allow us to approximately reach the target state by sequentially modifying the wavefunction. These heuristic idea is crucial for the proof of the following criterion by Boscain \etal
\begin{thm}\textnormal{\cite[Theorem 2.6]{BCCS}}\label{apxthm}
Let $c>0$ and let $(A, B, [0,c], \Phi_\mathcal{I})$ satisfy Assumption \ref{ass1}. If there exists a non-resonant chain of connectedness for $(A, B, [0,c], \Phi_\mathcal{I})$ then the system \eqref{contrsys} is approximately controllable.
\end{thm}
\subsection{Statement of the result}
In most of the physically relevant applications, the external control does not act on the spin part \cite{BMPS, Spohn}. 
Hence, we consider the JC dynamics with two different control terms acting only on the bosonic part, namely
\begin{equation}\label{cOP}
H_1=X\otimes\1\qquad H_2=P\otimes\1.
\end{equation}
To motivate our choice, we notice that -- for example -- in the cavity QED context the experimenters can only act on the electromagnetic field stored in the cavity. In this context the control terms $H_1,H_2$ correspond, respectively, to an external electric field and a magnetic field in the dipole approximation, see \cite[Section I.A]{BMPS}, and the control functions $u_1(t),u_2(t)$ model the amplitudes of this external fields.

With the previous choice, the complete controlled \Sch dynamics reads
\begin{equation}\label{cJCH}
\left\lbrace\begin{array}{l}
\iu\partial_t\psi=\Bigl(H_{\rm JC}(g)+u_1(t)H_1+u_2(t)H_2\Bigr)\psi \\[2mm]
\psi(0)=\Psi\sub{in} \in\Hi, \quad \Psi\sub{fin} \in\Hi\quad\mbox{s.t.}\;\norm{\Psi\sub{in}}=\norm{\Psi\sub{fin}}\\[2mm] 
u_1,u_2\in[0,c]\\[2 mm]
\omega,\Omega>0\\[2mm]
\abs{\Delta}\ll\omega,\Omega
\end{array}\right.
\end{equation}
Notice that the control functions $u_1,\,u_2$ are independent from each other so, as subcases, one can consider the system in which just one control is active. Obviously, controllability of the system in one of these two subcases implies controllability in the general case. This is exactly what we are going to prove. We consider the system \eqref{cJCH} in the subcases $u_1\equiv 0$ or $u_2\equiv 0$ and we prove that in each subcase the system is approximately controllable.

\medskip

The following theorems are the main results of the paper.
\begin{thm}[\bf Approximate controllability of JC dynamics]
\label{short}
The system \eqref{cJCH} with $u_1\equiv 0$ or $u_2\equiv 0$ is approximately controllable for 
every $g\in\R\setminus S_*$ where $S_*$ is a countable set.
\end{thm}  

\goodbreak

\begin{thm}[\bf Characterization of the singular set] \label{long}
The set $S_*$, mentioned in Theorem \ref{short}, consists of the value $g=0$ and those $g\in\R$ that satisfy one of the following equations:
\begin{eqnarray}
\label{criteig} E_{(n+1,-)}(g)=E_{(n,\nu)}(g), & (n,\nu)\in\mathcal{N}\\
\label{1c}  2\omega=f_{m+1}(g)+f_m(g)-f_{n+1}(g)+f_n(g), & n,m\in\mathfrak{N}_+\\
\label{1d}  2\omega=f_{m+1}(g)-f_m(g)-f_{n+1}(g)+f_n(g), & n,m\in\mathfrak{N}_-,\;m<n \qquad \\
\label{2c}  2\omega=f_{m+1}(g)+f_m(g)-f_{n+1}(g)-f_n(g), & n,m\in\mathfrak{N}_-,\;m>n  \qquad
\end{eqnarray}
where $\mathcal{N},\,\mathfrak{N}_\pm$ and $f_n$ are defined in \eqref{indset},\eqref{newN} and \eqref{fns}, respectively.
\end{thm}
The proof of Theorem \ref{short} follows two main steps: we introduce a Hilbert basis of eigenvectors of $H_{\rm JC}$, namely $\lbrace\ket{\bm{n}}\rbrace_{\bm{n}\in\mathcal{N}}$, and analyze the action of the control operators on it in order to show that all levels are coupled for every value of the parameter $g$ except a countable set (see Section \ref{step1}). We then construct a subset $\mathcal{C}_0$ of $\mathcal{N}^2$ and prove that it is a non-resonant chain of connectedness (see Section \ref{step2}). The claim then follows from the application of the general result by Boscain et al., namely Theorem\,\ref{apxthm}. 

To prove Theorem \ref{long} we carefully analyze the resonances of the system, which are solution to the forthcoming 
equation \eqref{res}. By proving that the latter has a countable number of solutions, we conclude that relevant pairs of energies are not resonant for every $g\in\R$ except the values in a countable set which will be characterized in the proof.

\section{Proof of Theorem \ref{short}}
\setcounter{subsection}{-1}
\subsection{Preliminaries}

Preliminarily, we have to show that Assumption\,\ref{ass1} is satisfied by
$$
(\iu H_{\rm JC}(g),\;\iu H_j,\;\R,\;\lbrace\ket{\bm{n}}\rbrace_{\bm{n}\in\mathcal{N}}),\quad\mbox{for}\;g\in\R \setminus S_0,\; 
j \in \set{1,2},
$$
where $S_0$ is a countable set.  Notice that the index set  $\mathcal{N}$ plays the role of the countable set $\mathcal{I}$ in Definition \ref{connchain}. 

We have already shown that $\lbrace\ket{\bm{n}}\rbrace_{\bm{n}\in\mathcal{N}}$ is a Hilbert basis of eigenfunctions for $H_{\rm JC}(g)$. Since $H_j$ is infinitesimally $H_0$-bounded ($H_j \ll H_0$ for short) for $j \in \set{1,2}$, then $H_j \ll H_{\rm JC}(g)$ 
for $j \in \set{1,2}$ (see \cite[Exercise XII.11]{RS4}). Hence \ref{item:A2} holds. Moreover, this implies that $H_{\rm JC}(g)+wH_j$ is selfadjoint on $\Do(H_{\rm JC})=\Do(H_0)$ for every $w\in\R$ (see \cite[Theorem X.12]{RS2}) and so \ref{item:A3} is satisfied.

As for assumption \ref{item:A4}, we observe that in view of the analyticity of the eigenvalues, there are just  countable many values of $g$ which correspond to eigenvalue intersections. With the only exception of these values, the eigenvalues are simple, so \ref{item:A4}  and Assumption \ref{ass1}  hold automatically for  every $g\in\R \setminus S_0$, where $S_0$ is the countable set corresponding to the eigenvalue intersections. 

On the other hand, we can further restrict the set of singular points from $S_0$ to $S_1 \subset S_0$. 
Indeed,  if two eigenvalues intersect in a point $g_*$, say $E_{\bm{n}}(g_*)=E_{\bm{m}}(g_*)$, property \ref{item:A4}  is still satisfied (by the same orthonormal system)
provided that $ \bra{\bm{m}}H_j \ket{\bm{n}}(g_*)=0,\; j \in \set{1,2}$.

Observe that, given $n\in\N$,
\begin{equation}\label{uncoup}
\abs{m -n} > 2  \quad\Rightarrow\quad\bra{\bm{m}}H_j \ket{\bm{n}}(g)=0\quad \forall g\in\R,\; j\in \set{1,2}. 
\end{equation}
Hence, {\it a priori} the only possibly problematic points are solutions to the following equations
\begin{equation}
E_{\bm{m}}(g)=E_{\bm{n}}(g)\qquad \bm{m},\bm{n}\in\mathcal{N},\quad {\bm{m}\neq\bm{n},}\quad \abs{m-n} \leq 2.
\end{equation}
By direct investigation, and using \eqref{Eq:Spectrum}, one notices that there are solutions only in the following cases 
(for $\bm{n} = {(n, \nu)} \in\mathcal{N}$):
\begin{align} \label{problematic} 
E_{\bm{n}}(g)=E_{(n+1,-)}(g) \quad 
\text{which is satisfied if and only if} \\ 
\nonumber
\abs{g}=G^{(1)}_{\bm{n}}:=\sqrt{\omega^2(2n+3)-\nu\sqrt{4\omega^4(n^2+3n+2)+\omega^2\Delta^2}};\\[5mm]
\label{innocent}E_{\bm{n}}(g)=E_{(n+2,-)}(g) \quad 
\text{which is satisfied if and only if} \\
\nonumber
\abs{g}=G^{(2)}_{\bm{n}}:=\sqrt{2\omega^2(n+2)-\nu\sqrt{4\omega^4(n^2+4n+3)+\omega^2\Delta^2}};\\[5mm]
\label{Eq:neutral}
{E_{(n,+)}(g)=E_{(n,-)}(g) \quad
\text{which is satisfied if and only if} }\\
\nonumber
{g=0\quad\mbox{and}\quad\Delta=0.}
\end{align}
We will establish {\it a posteriori} whether we have indeed to exclude those points by the analysis in the next subsection, after looking at the action of the control operators. 
\subsection{Coupling of energy levels}\label{step1}
To apply Theorem \ref{apxthm} to our case we need to build a non resonant chain of connectedness. As observed before in \eqref{uncoup}, the control operators do not couple most of the pairs. 

The coupling between remaining pairs is easily checked by using \eqref{low_raise}, \eqref{annih_cre},
 \eqref{eig_n+}, and \eqref{eig_n-}. For the sake of a shorter notation, we set $c_n:=\cos(\theta_n/2)$ and $s_n:=\sin(\theta_n/2)$.  Some straightforward  calculations for $H_1$ yield the following result: 
\begin{align*}
\bra{n,-}H_1\ket{n,+}=&\;0\\
\bra{n+1,+}H_1\ket{n,+}=&\;\frac{1}{\sqrt{2}}(\sqrt{n+1}c_nc_{n+1}+\sqrt{n+2}s_ns_{n+1})\neq 0 & \\
\bra{n+2,+}H_1\ket{n,+}=&\;0\\
\bra{n+1,-}H_1\ket{n,-}=&\;\frac{1}{\sqrt{2}}(\sqrt{n+1}c_nc_{n+1}+\sqrt{n+2}s_ns_{n+1})\neq 0 & \\
\bra{n+2,-}H_1\ket{n,-}=&\;0\\
\bra{n+1,-}H_1\ket{n,+}=&\;\frac{1}{\sqrt{2}}(\sqrt{n+2}s_nc_{n+1}-\sqrt{n+1}c_ns_{n+1})\neq 0 & \Leftrightarrow g\neq 0
\end{align*}  
\begin{align*}
\bra{n+2,-}H_1\ket{n,+}=&\;0 \\
\bra{n+1,+}H_1\ket{n,-}=&\;\frac{1}{\sqrt{2}}(\sqrt{n+2}c_ns_{n+1}-\sqrt{n+1}s_nc_{n+1})\neq 0 & \Leftrightarrow g\neq 0\\
\bra{n+2,+}H_1\ket{n,-}=&\;0\\
\bra{0,-}H_1\ket{-1,\delta}=&\;\frac{c_0}{\sqrt{2}}\geq\frac{1}{2}\\
\bra{0,+}H_1\ket{-1,\delta}=&\;\frac{s_0}{\sqrt{2}}\neq 0\Leftrightarrow g\neq 0.
\end{align*}  
From these computations we see that (compare with \eqref{problematic},\eqref{innocent}) in the points $\lbrace G^{(2)}_{\bm{n}}\rbrace_{\bm{n}\in\mathcal{N}}$ the system still satisfies Assumption \ref{ass1}, while in the points $\lbrace G^{(1)}_{\bm{n}}\rbrace_{\bm{n}\in\mathcal{N}}$ does not. {The point $g=0$ is never solution to \eqref{problematic} or \eqref{innocent} in view of the assumption $\abs{\Delta}\ll\omega$.  Moreover, since $\bra{n,-}H_1\ket{n,+}=0$ for every $g\in\R$, the system still satisfies Assumption \ref{ass1} for $g=0$, notwithstanding \eqref{Eq:neutral}.}

The same results hold for $H_2$. Moreover, in each of the previous cases one has 
\[ \bra{\bm{m}}H_2\ket{\bm{n}}=\iu \bra{\bm{m}}H_1\ket{\bm{n}}. \] 
We conclude that Assumption \ref{item:A4} is satisfied for every $g \in \R \setminus S_1$, 
where $S_1 := \lbrace G^{(1)}_{\bm{n}}\rbrace_{\bm{n}\in\mathcal{N}}$.

\goodbreak

\subsection{Non-resonances of relevant pairs}\label{step2}
Knowing exactly the pairs of levels coupled by the control terms,  
we claim that the set {(illustrated in Figure\,1)}
\begin{equation}\label{Czero}
\mathcal{C}_0=\left\lbrace \bigl[(n+1,+),(n,+)\bigr],\;\bigl[(n+1,+),(n,-)\bigr]\;|\;n\in\N\right\rbrace\cup\left\lbrace\bigl[(0,+),(-1,\delta)\bigr]\right\rbrace
\end{equation}
is a non-resonant chain of connectedness for every $g\in\R \setminus S_2$, where $S_2\subset\R$ is a countable set.


\begin{figure}
\label{Fig:1}
\caption{\small Schematic representation of the eigenstates of the JC Hamiltonian and the chain of connectedness $\mathcal{C}_0$, {in the case $\delta(\Delta)=-$.} Thick black lines correspond to pairs of eigenstates in the chain $\mathcal{C}_0$. Gray  dashed lines correspond instead to pairs of  eigenstates coupled by the control which are not in $\mathcal{C}_0$. \vspace{5mm}
} 
\centering
\begin{tikzpicture} [scale=.7]
         \draw [thick, dashed, lightgray] (-2,0)--(11,0);
         \draw [very thin] (0,2)--(9,2);
         \draw [thick, dashed] (9,2)--(11,2); 
         \draw (-2,-.4) node{$\ket{-1,\delta}$};
         \foreach\i in {0,...,4}{
          	\draw (2*\i,0)--++(0,-0.2);
          	\draw (2*\i,-.4) node{$\ket{\i,-}$};
          	\draw (2*\i,2)--++(0,0.2);
          	\draw (2*\i,2.4) node{$\ket{\i,+}$};
         }
         \foreach\i in {0,...,3}{
          	\draw [thick] (2*\i,2)--(2*\i+2,2);
          	\draw [thick] (2*\i,2)--(2*\i-2,0);	  
         }  
         \foreach\i in {0,...,4}{
         \draw [thick, dashed, lightgray] (2*\i,2)--(2*\i+2,0);
         } 
         \draw [thick] (8,2)--(6,0);
         \draw [thick] (9,1)--(8,0);
         \draw [thick, dashed] (10,2)--(9,1);
         
         \draw [thick] (8,2)--(9,2);
         \draw [thick, dashed] (9,2)--(10,2);        
\end{tikzpicture}
\end{figure}
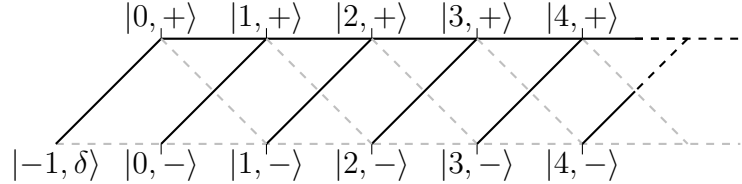


To prove this claim, we have to show that for every $g\in\R \setminus S_2$ 
each pair of eigenstates in $\mathcal{C}_0$ has no resonances with every other pair coupled by the control term. In view of the computation above, there are just four types of pairs coupled, as illustrated in Figure\,1 and 2. So, we define $S_2$ as the set of the solutions $g$ to the following equations:
\begin{equation}\label{res}
\abs{E_{\bm{k}}(g)-E_{\bm{l}}(g)}=\abs{E_{\bm{s}}(g)-E_{\bm{t}}(g)}
\end{equation}
where $[\bm{k},\bm{l}] \in\mathcal{C}_0$ and
\begin{align}
\nonumber
[\bm{s},\bm{t}] \in\mathcal{C}_0  
& \cup \left\lbrace \bigl[(n+1,-),(n,-)\bigr],\;\bigl[(n+1,-),(n,+)\bigr]\;|\;n\in\N\right\rbrace \\
&\cup  \left\lbrace\bigl[(0,-),(-1,\delta)\bigr]\right\rbrace.
\label{constraints} 
\end{align}
{It is enough to prove that the set of solutions to the latter equations is countable.}
Observe that, by the analyticity of the functions $g\mapsto E_{\bm{k}}(g)-E_{\bm{l}}(g)$, equation \eqref{res} may have at most countable many solutions unless is identically satisfied. Thus, we need to show that 
\[ E_{\bm{k}}(g)-E_{\bm{l}}(g)=\pm(E_{\bm{s}}(g)-E_{\bm{t}}(g))\]
is not satisfied for some value $g$ or, equivalently, that the Taylor expansions of r.h.s.\ and l.h.s.\ differ in at least a point. 
The same argument was used in \cite{BMPS}, where the authors computed the perturbative expansion of the eigenvalues of the Rabi Hamiltonian up to forth order in $g$. In our case, the model is exactly solvable,
\footnote{Actually, the Rabi model is also solvable, as recently proved by Braak \cite{Br}. However, the \virg{explicit} expression of the eigenvalues is practically intractable, so that the authors of \cite{BMPS} preferred to compute the perturbative series in $g$ up to fourth order.} 
so that we can compute the series expansions in $g=0$ directly from expression \eqref{eig}. 
An explicit computation yields the Taylor expansion:
$$\begin{array}{ll}
E_{\bm{n}}=\omega(n+1)+\nu\left(\frac{\abs{\Omega-\omega}}{2}+\frac{n+1}{\abs{\Omega-\omega}}g^2-\frac{(n+1)^2}{\abs{\Omega-\omega}^3}g^4+o\left(g^4\right)\right) & \quad\quad\mbox{for }\Delta\neq 0\\
E_{\bm{n}}=\omega(n+1)+\nu\sqrt{n+1}g & \quad\quad\mbox{for }\Delta=0.
\end{array}$$
It is now easy to check, mimicking  Step 2 in the proof of \cite{BMPS}, that for every choice of the indices in equation \eqref{res} the r.h.s.\ and l.h.s.\ have different series expansion at $g=0$. We are not going to detail this calculation, since in the next Section we will analyze in full detail equations \eqref{res}, in order to characterize the set $S_2$. By setting $S_* =S_1\cup S_2$, the proof of Theorem \ref{short} is concluded.

\section{Proof of Theorem \ref{long}}
\label{Sec:proof thm long}

In this proof we will discuss the resonances of the pairs of eigenstates in the chain $\mathcal{C}_0$, defined in 
\eqref{Czero}. Our aim is to provide conditions to determine whether, for a particular value of $g$, resonances are present or not. This requires a direct investigation of equation \eqref{res}. Since these computations are rather long, we prefer to collect them in this section, not to obscure the simplicity of the proof of Theorem \ref{short} with the details on the characterization of the set 
$S_*$. 

Let us recall that the sets $\mathfrak{N}_\pm$ defined in \eqref{newN} include both the natural numbers, and $-1$ is added to the one, among $\mathfrak{N}_+$ and $\mathfrak{N}_-$, whose label equals $\delta(\Delta)$. 

In each of the following subcases the existence of solutions to equation \eqref{res} is discussed, for every choice of the indices compatible with the constraints \eqref{constraints}. The mathematical arguments are based on elementary properties of functions $f_n(g)= \fhalf\sqrt{\Delta^2+4g^2(n+1)}$, which are summarized in Lemma \ref{fprop} in the Appendix. 

\medskip

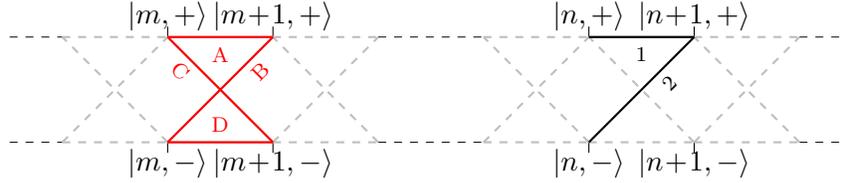
\begin{figure}
\caption{{\small Classification of the arcs representing pairs of eigenstates coupled by the control operator. On the left-hand panel, red arcs and labels show the classification of generic arcs in four different types $\set{A,B,C,D}$. 
On the right-hand panel, black arcs of the set $\mathcal{C}_0$ are classified in two types $\set{1,2}$.
Labels correspond to the cases enumerated in the proof of Theorem \ref{long}.} 
}
\vspace{2mm}
\centering
\begin{tikzpicture} [scale=.7]\label{resograph}
		 \draw [very thin,dashed] (-3,0)--(-2,0);
         \draw [very thin,dashed] (4,0)--(6,0);
         \draw [very thin,dashed] (12,0)--(13,0);
         \draw [very thin,dashed] (-3,2)--(-2,2);
         \draw [very thin,dashed] (4,2)--(6,2);
         \draw [very thin, dashed] (12,2)--(13,2);
         \foreach\i in {0,1,4,5}{
          	\draw (2*\i,2)--++(0,0.2);
         }
         \foreach\i in {0,1,4,5}{
          	\draw (2*\i,0)--++(0,-0.2);
         }
         \foreach\i in {0}{ 
          	\draw [thick, lightgray] (2*\i,0)--(2*\i+2,2);}
         \foreach\i in {0}{ 
          	\draw [thick, lightgray] (2*\i+2,0)--(2*\i,2);}
         \foreach\i in {-1,1,3,4,5}{ 
          	\draw [thick, dashed, lightgray] (2*\i+2,0)--(2*\i,2);}
         \foreach\i in {-1,1,3,5}{ 
          	\draw [thick, dashed, lightgray] (2*\i,0)--(2*\i+2,2);}
         \foreach\i in {0,4}{ 
          	\draw [thick, lightgray] (2*\i,2)--(2*\i+2,2);
         }  
         \foreach\i in {-1,1,3,5}{ 
          	\draw [thick, dashed, lightgray] (2*\i,2)--(2*\i+2,2);
         }   
         \foreach\i in {0}{ 
          	\draw [thick, lightgray] (2*\i,0)--(2*\i+2,0);
         }
         \foreach\i in {-1,1,3,4,5}{ 
          	\draw [thick, dashed, lightgray] (2*\i,0)--(2*\i+2,0);
         }
         \foreach\i in {}{ 
         	\draw [thick, lightgray] (2*\i+1,.375) to [out=0,in=140] (2*\i+2,0);}
         \foreach\i in {}{ 
         	\draw [thick, lightgray] (2*\i,2) to [out=-40,in=180] (2*\i+1,1.625);}
         \foreach\i in {}{ 
         	\draw [thick, dashed, lightgray] (2*\i+1,1.625) to [out=0,in=-140] (2*\i+2,2);}
         \foreach\i in {}{ 
         	\draw [thick, dashed, lightgray] (2*\i+1,1)--(2*\i+2,2);}
         \foreach\i in {}{ 
         	\draw [thick, lightgray] (2*\i,0)--(2*\i+1,1);}
         \draw [thick] (8,2)--(10,2);
         \node [below] at (9,2) {\tiny 1};
                  	
         \draw [thick] (8,0)--(10,2);
         \node [label={[rotate=45]right:\tiny 2}] at (9.1,.9) {};
         	\draw [thick, red] (0,2)--(2,2);
         	\node [below, red] at (1,2) {\tiny A};
         	\draw [thick, red] (0,0)--(2,0);
         	\node [above, red] at (1,0) {\tiny D};
         	\draw [thick, red] (0,0)--(2,2);
         	\node [label={[red, rotate=45]right:\tiny B}] at (1.3,1.1) {};
         	\draw [thick, red] (2,0)--(0,2);
         	\node [label={[red, rotate=-45]left:\tiny C}] at (.7,1.1) {};
        
         \draw (0,2.4) node{\small $\ket{m,+}$};
         \draw (2,2.4) node{\small $\ket{m\!+\!1,+}$};
         \draw (8,2.4) node{\small $\ket{n,+}$};
         \draw (10,2.4) node{\small $\ket{n\!+\!1,+}$};
         \draw (0,-.4) node{\small $\ket{m,-}$};
         \draw (2,-.4) node{\small $\ket{m\!+\!1,-}$};
         \draw (8,-.4) node{\small $\ket{n,-}$};
         \draw (10,-.4) node{\small $\ket{n\!+\!1,-}$};
        
\end{tikzpicture}
\end{figure}

{\bf Case 1:  Assume $[\bm{k},\bm{l}]={[(n +1,+),(n,+)]},\;n\in\mathfrak{N}_+$.} 
This assumption correspond to select the black arc labeled by 1 in the graph in Figure \ref{resograph}. We investigate the possible resonances between the selected arc and the other arcs, classified according to their qualitative type (see the labels in the left-hand panel of Figure \ref{resograph}). This analysis amounts to consider, with the help of Lemma \ref{fprop} (properties
\ref{item:i}-\ref{item:iv}), the following subcases:
\begin{enumerate}[itemsep=2ex,label={$({1.\mathrm{\Alph*}})$},ref=$({1.\mathrm{\Alph*}})$]
\item \label{item:1A}
$[\bm{s},\bm{t}]=\bigl[(m+1,+),(m,+)\bigr],\;m\in\mathfrak{N}_+,\;m\neq n$. 
		Equation \eqref{res} reads
		\[ f_{n+1}(g)-f_n(g)=f_{m+1}(g)-f_m(g)\]
		which is satisfied if and only if $g=0$, because $f_{n+1}(g)-f_n(g)$ is strictly decreasing in $n$
		for $g \neq 0$ in view of \ref{item:iv}. 
\item \label{item:1B}		
		$[\bm{s},\bm{t}]=\bigl[(m+1,+),(m,-)\bigr],\;m\in\mathfrak{N}_-.$
		Equation \eqref{res} reads
		\[ f_{n+1}(g)-f_n(g)=f_{m+1}(g)+f_m(g)\]
		which is satisfied if and only if $g=0$ and $\Delta=0$. 
		\newline
		Indeed, by \ref{item:ii} one has
		\begin{align*}
		\qquad f_{n+1}(g)-f_n(g) &\leq {2}\abs{g} \(\sqrt{n+2}-\sqrt{n+1}\) = \frac{2 \abs{g}}{\sqrt{n+2} + \sqrt{n+1}}. 
		\end{align*} 
		For $\Delta< 0$, one has $n \in\mathfrak{N}_+ = \N$ and $m\in\mathfrak{N}_-= \N \cup \{-1\}$, 
		as illustrated in Figure\,1. Hence, for $g \neq 0$,
		\begin{align*}
		\qquad f_{n+1}(g)-f_n(g) &\leq \frac{2 \abs{g}}{\sqrt{2} + 1} < \abs{g}  
		                          \leq \abs{g} (\sqrt{m+2}+\sqrt{m+1}) \\
		                          & {\leq} \, f_{m+1}(g) {+} f_m(g),                  
		\end{align*} 
		and the last inequality is strict whenever $\Delta\neq 0$.  \newline      
		Analogously, for $\Delta > 0$ one has $n \in\mathfrak{N}_+ = \N \cup  \{-1\}$ and $m\in\mathfrak{N}_-= \N$. 
		Hence, for $g \neq 0$,
		\begin{align*}
		\qquad f_{n+1}(g)-f_n(g) &\leq  2 \abs{g}  
					 <  \abs{g} (\sqrt{m+2}+\sqrt{m+1}) \\
		                          & {\leq} \, f_{m+1}(g) {+} f_m(g).
		\end{align*} 		
		As above, the last inequality is strict whenever $\Delta\neq 0$.
\item \label{item:1C}		
		$[\bm{s},\bm{t}]=\bigl[(m+1,-),(m,+)\bigr],\;m\in\mathfrak{N}_+.$
		Equation \eqref{res} reads 
		\[ \omega+f_{n+1}(g)-f_n(g)=\abs{\omega-f_{m+1}(g)-f_m(g)}.\]
		If $\abs{g}<G^{(1)}_{m,+}$ one has
		\[ f_{n+1}(g)-f_n(g)=-f_{m+1}(g)-f_m(g),\]
		which implies $g=0$ and $\Delta=0$ because $-f_{m+1}(g)-f_m(g)\leq 0\leq f_{n+1}(g)-f_n(g)$, 
		{and the first inequality is strict whenever $\Delta\neq 0$, while the second inequality is strict whenever $g\neq 0$.} 
		\newline
		On the other hand, if $\abs{g}\geq G^{(1)}_{m,+}$ the equation above reads
		\[ 2\omega=f_{m+1}(g)+f_m(g)-f_{n+1}(g)+f_n(g)\qquad \eqref{1c}\]
		which has two solutions because the r.h.s.\ is equal to $\abs{\Delta}$ in zero (and $\abs{\Delta} \ll \omega$ in view of   		\eqref{cJCH}) and is strictly increasing in $\abs{g}$. Indeed, one easily sees that
		\begin{align*}
		\partial_g&\left(f_{m+1}(g)+f_m(g)-f_{n+1}(g)+f_n(g)\right)=\\
		&\quad g\left(\frac{m+2}{f_{m+2}(g)}+\frac{m+1}{f_{m}(g)}-\frac{n+2}{f_{n+1}(g)}+\frac{n+1}{f_n(g)}\right)=:gC_{m,n}(g)
		\end{align*}
		{where $C_{m,n}(g)>0$ for every choice of {indices} $n,m\in\mathcal{N}_+$ and $\Delta\neq 0$. For $\Delta=0$ the r.h.s. of \eqref{1c} is $\abs{g}(\sqrt{m+2}+\sqrt{m+1}-\sqrt{n+2}+\sqrt{n+1})$ which is clearly strictly increasing in $\abs{g}$}.
\item \label{item:1D}				
		$[\bm{s},\bm{t}]=\bigl[(m+1,-),(m,-)\bigr],\;m\in\mathfrak{N}_-.$
		Equation \eqref{res} reads 
		\[ \omega+f_{n+1}(g)-f_n(g)=\abs{\omega-f_{m+1}(g)+f_m(g)} \]
		Then if $\abs{g}<G^{(1)}_{m,-}$
		\[ f_{n+1}(g)-f_n(g)=-f_{m+1}(g)+f_m(g), \]
		which is satisfied if and only if $g=0$ because $f_{n+1}(g)-f_n(g)\geq 0\geq -f_{m+1}(g)+f_m(g)$
		and the inequalities are strict whenever $g\neq 0$. 
		If instead $\abs{g} \geq G^{(1)}_{m,-}$, the equation reads 		
		\[ 2\omega=f_{m+1}(g)-f_m(g)-f_{n+1}(g)+f_n(g)\qquad \eqref{1d} \]
		which has two solutions if and only if $m<n$. Indeed, $f_{n+1}(g)-f_n(g)$ is decreasing in $n$ 
		in view of \ref{item:iv} {and the derivative of the r.h.s. is}
		\begin{align*}
		\partial_g&\left(f_{m+1}(g)-f_m(g)-f_{n+1}(g)+f_n(g)\right)=\\
		&\quad g\left(\frac{m+2}{f_{m+2}(g)}-\frac{m+1}{f_{m}(g)}-\frac{n+2}{f_{n+1}(g)}+\frac{n+1}{f_n(g)}\right)=:gD_{m,n}(g).
		\end{align*}
		{The function $D_{m,n}(g)$ is strictly positive for every $\Delta\neq 0$ and $m\in\mathcal{N}_-,\,n\in\mathcal{N}_+$ with $m<n$ because $\frac{n+2}{f_{n+1}(g)}-\frac{n+1}{f_n(g)}$ is strictly decreasing in $n$. For $\Delta=0$ the r.h.s.\ of \eqref{1d} is $\abs{g}(\sqrt{m+2}-\sqrt{m+1}-\sqrt{n+2}+\sqrt{n+1})$ which is positive if and only if $m<n$, and is clearly strictly increasing in $\abs{g}$.}
\end{enumerate}
In view of the above analysis, there exist 
non trivial resonances {(for $g\neq 0$)} in cases \ref{item:1C}, and \ref{item:1D} for $m<n$. 
In such circumstances, equations \eqref{1c},\eqref{1d} have two solutions each.{As for the trivial value $g=0$, the system exhibits multiple resonances, as noted in all previous cases. Hence, $g=0$ has to be included in the set  of resonant points.}

\medskip

{\bf Case 2:  Assume $[\bm{k},\bm{l}]=[(n+1,+),(n,-)],\;n\in\mathfrak{N}_-$.} 
This assumption corresponds to select the black arc labeled by 2 in the graph in Figure \ref{resograph}.
 As before, we proceed by considering the following sub-cases: 
\begin{enumerate}[itemsep=2ex,label={$({2.\mathrm{\Alph*}})$},ref=$({2.\mathrm{\Alph*}})$]
\item \label{item:2A} By symmetry, this case reduces to the subcase \ref{item:1B}. As already noticed, a
 solution exists if and only if $g=0$ and $\Delta=0$.
\item \label{item:2B}
		$[\bm{s},\bm{t}]=\bigl[(m+1,+),(m,-)\bigr],\;m\in\mathfrak{N}_-,\;m\neq n.$ The corresponding equation reads
		\[ f_{n+1}(g)+f_n(g)=f_{m+1}(g)+f_m(g)\] 
		which has only the trivial solution $g=0$. 
\item \label{item:2C}
		$[\bm{s},\bm{t}]=\bigl[(m+1,-),(m,+)\bigr],\;m\in\mathfrak{N}_+.$ Equation \eqref{res} reads
		\[ \omega+f_{n+1}(g)+f_n(g)=\abs{\omega-f_{m+1}(g)-f_m(g)}.\]
		Then if $\abs{g}<G^1_{m,+}$ the equation above become
		\[f_{n+1}(g)+f_n(g)=-f_{m+1}(g)-f_m(g)\] 
		which clearly implies that $g=0$ and $\Delta=0$. \newline
		On the other hand, if $\abs{g}\geq G^1_{m,+}$ the equation reads
		\[ 2\omega=f_{m+1}(g)+f_m(g)-f_{n+1}(g)-f_n(g)\qquad \eqref{2c}\]
		which has non trivial solutions if and only if $m>n$ because $f_n$ is increasing in $n$ for $g\neq 0$. Since the r.h.s. is strictly increasing in $\abs{g}$, as one can see using an argument similar to case \ref{item:1C}, the latter equation has two solutions if $m>n$. 
\item \label{item:2D}
		$[\bm{s},\bm{t}]=\bigl[(m+1,-),(m,-)\bigr],\;m\in\mathfrak{N}_-.$  Equation \eqref{res} reads
		\[ \omega+f_{n+1}(g)+f_n(g)=\abs{\omega-f_{m+1}(g)+f_m(g)}. \]
		If $\abs{g}<G^1_{m,-}$ the above equation reads 
		\[f_{n+1}(g)+f_n(g)=-f_{m+1}(g)+f_m(g)\]
		which is satisfied if and only if  $g=0$ and $\Delta=0$, since $-f_{m+1}(g)+f_m(g)\leq 0\leq f_{n+1}(g)+f_n(g)$, 
		{and the first inequality is strict whenever $g\neq 0$, while the second inequality is strict whenever $\Delta\neq 0$}. 
		If, instead, $\abs{g}\geq G^1_{m,-}$, the above equation becomes
		\[ 2\omega=f_{m+1}(g)-f_m(g)-f_{n+1}(g)-f_n(g), \]
		which has no solution since the r.h.s.\ is non-positive for every $g,\Delta\in\R$ and  $n,m\in\mathcal{N}_-$.
\end{enumerate}

In summary, as far as Case 2 is concerned, there exist non trivial resonances (for $g \neq 0$) only in the case \ref{item:2C} for $m>n$, and in such case equation \eqref{2c} has exactly two solutions. 

Recalling the definition of $\mathcal{C}_0$ (see \eqref{Czero}), one notices that every element 
of $\mathcal{C}_0$ is non-resonant with every other element of $\mathcal{C}_0$ except for the trivial value $g=0$, 
in view of the analysis of the cases \ref{item:1A}, \ref{item:1B},  \ref{item:2A} and \ref{item:2B}.

The proof above exhibits equations \eqref{1c}, \eqref{1d} and \eqref{2c} appearing in the statement  of Theorem \ref{long}, 
as the equations which characterize the values of $g$ in $S_2$, namely those values such that an arc in $\mathcal{C}_0$ 
is resonant with some arc {(not in $\mathcal{C}_0$)} non-trivially coupled by the interaction.   
As we said before, the value $g=0$ is included in $S_2$. 

Finally, one has to include in $S_*$ those values of $g$ for which {some eigenspace has dimension $2$ and the corresponding eigenvectors are coupled}. These values, {defining the set $S_1$, 
have been already characterized by equation \eqref{criteig}, whose solutions are exhibited in \eqref{problematic}. 
} 

In view of Theorem \ref{apxthm}, we conclude that the system is approximately controllable for every $g\in\R\setminus S_*$, where $S_*= S_1\cup S_2$ is characterized by equations \eqref{criteig}, \eqref{1c}, \eqref{1d}, and \eqref{2c}.
This concludes the proof of Theorem \ref{long}.

\bigskip


\begin{appendix}

\section{Ancillary results}

The following Lemma contains a list of useful elementary properties of the functions $\{f_n\}_{n\in\N}$, 
which have been used in the proof of Theorem \ref{long} (Section \ref{Sec:proof thm long}).

\begin{lemma}\label{fprop} Let $f_n$, $n\in\N$, be defined as in \eqref{fns}. Then 
\begin{enumerate}[label={$(\mathrm{F}.{\arabic*})$},ref=$(\mathrm{F}.{\arabic*})$]
\item \label{item:i}$f_m(g)-f_n(g)\geq 0\quad \text{ if and only if }\quad m\geq n$; \vspace{2mm}
\item \label{item:ii}$f_{n+1}(g)-f_n(g)\leq {2}\abs{g}(\sqrt{n+2}-\sqrt{n+1})$; \vspace{2mm}
\item \label{item:iv}$f_{n+1}(g)-f_n(g)$ is strictly increasing w.r.to $\abs{g}$, and strictly decreasing in $n$; \vspace{2mm}
\end{enumerate}
\end{lemma}

\begin{proof} 
Property \ref{item:i} follows from the monotonicity of the square root.
As for \ref{item:ii}, one notices that
$$
\quad f_{n+1}(g)-f_n(g)\leq \frac{{2}g^2}{\sqrt{\Delta^2+4g^2(n+2)}}
$$
which is equivalent to 
$$
(\Delta^2+4g^2(n+2))-\sqrt{\Delta^2+4g^2(n+2)}\sqrt{\Delta^2+4g^2(n+1)}\leq 
{4}g^2
$$
which follows from the fact that 
$$
(\Delta^2+4g^2(n 
+{1}))\leq\sqrt{\Delta^2+4g^2(n+2)}\sqrt{\Delta^2+4g^2(n+1)}. 
$$
Then, 
\begin{align*} \label{}
\quad f_{n+1}(g)-f_n(g) 
&\leq \frac{{2}g^2}{\sqrt{\Delta^2+4g^2(n+2)}}
\leq\frac{{2}g^2}{2\abs{g}\sqrt{n+2}} \\
& \leq {2}\abs{g}(\sqrt{n+2}-\sqrt{n+1}). 
\end{align*}
Notice that the last inequality is strict whenever $g\neq 0$.


As for \ref{item:iv}, one sets 
$$
F(x,y):=\frac{1}{2}\sqrt{\Delta^2+4x^2y} \quad \mbox{ and } \quad  G(x,y):=y/\sqrt{\Delta^2+4x^2y},
$$ 
so that $f_n(g)=F(g,n+1)$ and $\partial_g f_n(g)=2g \, G(g,n+1)$.  
Observe that for $y\geq 0$ one has 
\begin{align*}
\frac{\partial G}{\partial y} 
=\frac{1}{\sqrt{\Delta^2+4x^2y}}\left(1-\frac{2x^2y}{\Delta^2+4x^2y}\right) > 0
\end{align*}
and also
\begin{align*}
\frac{\partial^2 G}{\partial y^2}=
&\frac{4x^2}{(\Delta^2+4x^2y)^{3/2}}\left(-1+\frac{3y}{4}\frac{4x^2}{\Delta^2+4x^2y}\right)
\leq 0,
\end{align*}
and the latter is equal to $0$ if and only if $x=0$.
Then, one has (with an innocent abuse of notation concerning $\partial_n f_{n}(g)$)
\begin{align*}
\partial_g(f_{n+1}-f_n)(g)=&2g(G(g,n+2)-G(g,n+1))\quad\left\lbrace\begin{array}{rl}>0& \mbox{for } g>0\\<0& \mbox{for } g<0\end{array}\right. ; \\
\partial_n(f_{n+1}-f_n)(g)=&\frac{\partial F}{\partial y}(g,n+2)-\frac{\partial F}{\partial y}(g,n+1)\\
=&g^2\left(\frac{1}{\sqrt{\Delta^2+4g^2(n+2)}}-\frac{1}{\sqrt{\Delta^2+4g^2(n+1)}}\right)<0.
\end{align*}
The monotonicity properties claimed in the statement follow immediately. 

%
\end{proof}

\end{appendix}



\vspace{1cm}

{\footnotesize  

\begin{tabular}{ll}

(L. Pinna) & \textsc{Dipartimento di Matematica,  \virg{La Sapienza} Universit\`{a} di Roma} \\
 &  Piazzale Aldo Moro 2, 00185 Rome, Italy \\
 &  {E-mail address}: \href{mailto:pinna@mat.uniroma1.it}{\texttt{pinna@mat.uniroma1.it}} \\
 \\
(G. Panati) 
&  \textsc{Dipartimento di Matematica, \virg{La Sapienza} Universit\`{a} di Roma} \\
 &  Piazzale Aldo Moro 2, 00185 Rome, Italy \\
 &  {E-mail address}: \href{mailto:panati@mat.uniroma1.it}{\texttt{panati@mat.uniroma1.it}} \\

\end{tabular}

}

\begin{thebibliography}{999999}

\bibitem[BMS]{BMS} \textsc{Ball, J.M.; Marsden, J.E.; Slemrod, M.}: Controllability for
distributed bilinear systems, {\it SIAM J. Control Optim.} {\bf 20}(4):  575--597(1982).

\bibitem[BCCS]{BCCS}
\textsc{Boscain, U.; Caponigro, M.; Chambrion, T.; Sigalotti, M.}: A Weak Spectral Condition for the Controllability of the Bilinear Schr\"{o}dinger Equation with Application to the Control of a Rotating Planar Molecule, {\it Commun. Math. Phys.} {\bf 311}:  423--455 (2012).

\bibitem[BMPS]{BMPS}
\textsc{Boscain, U.; Mason, P.; Panati, G.; Sigalotti, M.} : On the control of spin-boson systems, {\it J. Math. Phys.} {\bf 56}: 092101 (2015). 

\bibitem[Br]{Br}
\textsc{Braak, D.} : Integrability of the Rabi model, {\it Phys. Rev. Lett.} {\bf 107}: 100401 (2011).

\bibitem[BRH]{BRH}
\textsc{Brune, M.; Haroche, S.; Raimond, J.M.} : Colloquium: Manipulating quantum entanglement with atoms and photons in a cavity, {\it Rev. Mod. Phys.} {\bf 73}: 565--582 (2001).

\bibitem[Co$_1$]{Coh1}
\textsc{Cohen-Tannoudji, C.; Dupont-Roc, J.; Grynberg, G.} : {\it Photons\&Atoms: Introduction to Quantum Electrodynamics}, Wiley Professional Paperback Series, John Wiley \& Sons Inc., New York, 1989.

\bibitem[Co$_2$]{Coh2} \textsc{Cohen-Tannoudji C.; Diu B.; Lalo\"e F.}: {\it Quantum Mechanics}, Wiley, 1977.

\bibitem[EbLa]{EbLaw}
\textsc{Eberly, J.H.; Law, C.K.} : Arbitrary control of a quantum electro-magnetic field, {\it Phys. Rev. Lett.} {\bf 76}: 1055--1058 (1996).

\bibitem[ErPu]{ErvPuel}
\textsc{Ervedoza, S.; Puel, J.P.} : Approximate controllability for a system of Schr\"{o}dinger equation modeling a single trapped ion, {\it Ann. Inst. H. Poincar\'{e} Anal. Non Lin\'{e}air} {\bf 26}: 2111--2136 (2008).

\bibitem[HaRa]{HaRa}
\textsc{Haroche, S.; Raimond, J.M.} : {\it Exploring the Quantum. Atoms, Cavities and Photons}, Oxford Graduate Texts, Oxford University Press, New York, 2006.

\bibitem[Ka]{Ka}
\textsc{Kato, T.} : {\it Perturbation Theory for Linear Operator}, Die Grundlehren der mathematischen Wissenschaften, vol.132, Springer-Verlag, New York, 1966.

\bibitem[JaCu]{JC}
\textsc{Jaynes, E.T.; Cummings, F.W.} : Comparison of quantum and semiclassical radiation theories with application to the beam maser, {\it Proceedings of the IEEE} {\bf 51}: 89--109 (1963).

\bibitem[MiRo]{uncontr}
\textsc{Mirrahimi, M.; Rouchon, P.} : Controllability of Quantum Harmonic Oscillators, {\it IEEE Trans Automatic Control} {\bf 49}: 745--747 (2004).

\bibitem[PaSi]{PaduroSigalotti}
\textsc{Paduro, E.; Sigalotti, M.} : Approximate controllability of the two trapped ions system.
{\it Quantum Inf Process} {\bf 14}: 2397--2418 (2015).




\bibitem[PSpT]{PSparberT}
\textsc{Panati, G.;  Sparber, Ch.; Teufel, S.}: 
Geometric currents in piezoelectricity,
{\it Arch. Rat. Mech. Analysis} {\bf 191}:  387--422 (2009).

\bibitem[PST$_1$]{PST1}
\textsc{Panati, G.;  Spohn, H.; Teufel, S.}: 
Motion of electrons in adiabatically perturbed periodic structures.
In: A. Mielke (Ed.), {\it Analysis, Modeling and Simulation of Multiscale Problems}, Springer, 2006. 

\bibitem[PST$_2$]{PST2}
\textsc{Panati, G.;  Spohn, H.; Teufel, S.}: 
The time-dependent Born-Oppenheimer approximation, 
{\it Mathematical Modelling and Numerical Analysis}  {\bf 41}: 297--314 (2007).

\bibitem[Ra$_1$]{Rab1}
\textsc{Rabi I.I.}: On the process of space quantization, {\it Phys. Rev} {\bf 49}: 324 (1936).

\bibitem[Ra$_2$]{Rab2}
\textsc{Rabi I.I.}: Space quantization in a gyrating magnetic field, {\it Phys. Rev} {\bf 51}: 652 (1937).

\bibitem[RS$_2$]{RS2}
\textsc{Reed, M.; Simon, B.}: {\it Methods of modern mathematical physics. Volume II: Fourier analysis, self-adjointness}, Academic Press, New York, 1975.

\bibitem[RS$_4$]{RS4}
\textsc{Reed, M.; Simon, B.}: {\it Methods of modern mathematical physics. Volume IV: Analysis of operators}, Academic Press, New York, 1978.

\bibitem[Ro]{Rou}
\textsc{Rouchon, P.} : Quantum systems and control, {\it ARIMA Rev. Afr. Rech. Inform. Math Appl.} {\bf 9}: 325--357(2008).

\bibitem[Sp]{Spohn}
\textsc{Spohn, H.} : {\it Dynamics of charged particles and their radiation field}, Cambridge University Press, 2004.

\bibitem[Tur]{Tur}
\textsc{Turinici, G.} : On the controllability of bilinear quantum systems. 
In:  M. Defranceschi and C. Le Bris (Eds.), {\it Mathematical models and methods for ab initio Quantum Chemistry},
Lecture Notes in Chemistry {\bf 76}, Springer, 2000.

\end{thebibliography}
\end{document}